\documentclass[english,aps,manuscript]{revtex4}
\usepackage[T1]{fontenc}
\usepackage[latin9]{inputenc}
\usepackage[letterpaper]{geometry}
\usepackage{endnotes}
\usepackage{amsthm}
\usepackage{amsmath}

\makeatletter
\@ifundefined{textcolor}{}
{%
 \definecolor{BLACK}{gray}{0}
 \definecolor{WHITE}{gray}{1}
 \definecolor{RED}{rgb}{1,0,0}
 \definecolor{GREEN}{rgb}{0,1,0}
 \definecolor{BLUE}{rgb}{0,0,1}
 \definecolor{CYAN}{cmyk}{1,0,0,0}
 \definecolor{MAGENTA}{cmyk}{0,1,0,0}
 \definecolor{YELLOW}{cmyk}{0,0,1,0}
 }
\theoremstyle{plain}
 \let\footnote=\endnote

\makeatother

\newtheorem{thm}{Theorem}  
 \newtheorem{lem}{Lemma}
\newtheorem{prop}{Proposition} 
\newtheorem{remark}{Remark}

\makeatother

\usepackage{babel}

\makeatother

\usepackage{babel}

\begin{document}

\title{More nonlocality with less purity}

\author{Somshubhro Bandyopadhyay}

\address{Department of Physics and Center for Astroparticle Physics and Space
Science, Bose Institute, Block EN, Sector V, Salt Lake, Kolkata 700091,
India}

\email{som@bosemain.boseinst.ac.in}
\begin{abstract}
Quantum information is nonlocal in the sense that local measurements
on a composite quantum system, prepared in one of many mutually orthogonal
states, may not reveal in which state the system was prepared. It
is shown that in the many copy limit this kind of nonlocality is fundamentally
different for pure and mixed quantum states. In particular, orthogonal
mixed states may not be distinguishable by local operations and classical
communication (LOCC), no matter how many copies are supplied, whereas
any set of \emph{N} orthogonal pure states can be perfectly discriminated
with $m$ copies, where, $m<N$. Thus mixed quantum states can exhibit
a new kind of nonlocality absent in pure states. We also argue that
a set of orthogonal quantum states may be said to be maximally indistinguishable
iff the set is not conclusively locally distinguishable with multiple
copies. 
\end{abstract}
\maketitle
One of the primary goals in quantum information theory is to understand
the relationship between nonlocality of quantum information and entanglement.
The nonlocality of quantum information is in the sense that joint
measurements on a composite quantum system can reveal more about the
state than by local operations and classical communication (LOCC)
alone. The problem of local distinguishability of quantum states \cite{Bandyo+Walgate-2009,Bandyopadhyay2010,Bennett-I-99,Bennett-II-99,Chefles2004,Chen+Li,ChenYang2002,Cohen2007,Duan2007,DuanFeng2007,Fan-2004,Ghosh-2001,Ghosh-2002,Ghosh-2004,Hayashi-etal-2006,Horodecki-2003,Nathanson-2005,Virmanietal2001,Walgate-2000,Walgate-2002,Watrous-2005,Bagan-2010}
provides the setting in which questions on this kind of nonlocality
and entanglement are generally explored. 

The set up of local distinguishability of quantum states is simple:
Two or more spatially separated observers share a composite quantum
system prepared in one of many mutually orthogonal quantum states.
Their goal is to identify the state in which the system was prepared
within the constraints of LOCC; that is, they can perform any sequence
of coordinated measurements on their respective subsystems but are
not allowed to exchange quantum states. In some cases they can indeed
accomplish this task without error and in some they cannot. For example,
any two pure orthogonal states can be perfectly distinguished \cite{Walgate-2000}
whereas, it is impossible to do so if the unknown state belongs to
a complete orthogonal basis in which one or more states are entangled
\cite{Horodecki-2003}. 

While entanglement has been known to ensure difficulty in state discrimination
\cite{Hayashi-etal-2006} and is commonplace in most LOCC indistinguishable
sets of quantum states, it is however, not necessary. The product
states exhibiting {}``non-locality without entanglement'' or forming
an unextendible product basis (UPB), despite being mutually orthogonal,
cannot be perfectly distinguished by LOCC \cite{Bennett-I-99,Bennett-II-99}.
In all the aforementioned examples complete information about the
state can only be obtained by measuring the whole system rather than
local measurements of its parts, and hence these states are said to
exhibit nonlocality. 

A central feature of LOCC discrimination of quantum states is that
the observers hold a single copy of the unknown state, and under this
{}``single copy'' constraint the nonlocal nature of quantum information
is manifested. Suppose we relax this assumption and allow many copies
of the unknown state to be shared between the participants, then reliable
identification of the given state may be possible. For example, the
canonical Bell basis in $d\otimes d$ can be perfectly distinguished
with only two copies \cite{Ghosh-2004} while it is indistinguishable
for a single copy \cite{Ghosh-2004,Fan-2004,Nathanson-2005}. Unfortunately,
a general result on the number of copies sufficient to reliably distinguish
an arbitrary orthogonal ensemble is not yet known. Nevertheless, it
seems reasonable to conjecture that if sufficiently many copies of
the unknown state are available, one can perfectly distinguish any
set of orthogonal quantum states. Consequently, the nonlocality of
quantum information that we are talking about may not be observed
in the many copy limit. 

It is shown that this is indeed the case for pure states but does
not hold in general for mixed states. We show that any set of $N$
orthogonal pure states can be perfectly distinguished by LOCC while
requiring no more than $\left(N-1\right)$ copies. This bound is perhaps
not tight and better bounds may be obtained by invoking clever measurement
strategies. Nonetheless it serves our purpose to demonstrate that
only finitely many copies are required to distinguish pure states
exactly.

On the other hand we show that a set of orthogonal quantum states,
not all of which are pure, may not be perfectly distinguished by LOCC
even if $n$ copies of the unknown state are available, where $n$
is finite but arbitrarily large. We explicitly construct examples
of such sets of states from unextendible product basis \cite{Bennett-II-99,Divincenzo-2003}.
Surprisingly such a set can be minimal, consisting of only two orthogonal
quantum states, one of which is necessarily mixed. Thus, mixed quantum
states can exhibit nonlocality in the many copy limit, a feature which
pure states do not possess. In some sense, this nonlocality, which
is similar to that of {}``nonlocality without entanglement'' \cite{Bennett-I-99}
and {}``more nonlocality with less entanglement'' \cite{Horodecki-2003},
is more robust in mixed states for it persists even in the domain
of multiple copies, whereas in case of pure states it does not. 

As noted before, the essence of this kind of quantum nonlocality lies
in the fact that it is not always possible to access the {}``which
state'' information locally even though the system was known to be
in one of several mutually orthogonal quantum states. Our result shows
that, in the many copy limit, a set of orthogonal states that cannot
be perfectly distinguished by LOCC is therefore relatively more nonlocal
than a set that can be perfectly distinguished. We now go a step further
and ask the following question: what is the strongest form of nonlocality
manifested in the setting of LOCC discrimination of orthogonal quantum
states? We answer this question by laying down a simple, yet physically
meaningful criterion and give examples of orthogonal quantum states
that satisfy it. 

What we propose is intuitively easy to understand and relies on the
notion of conclusive local distinguishability \cite{Duan2007,Nathanson-2005,Chefles2004,Bandyo+Walgate-2009,Virmanietal2001,ChenYang2002}.
Recall that in perfect LOCC discrimination, the {}``which state''
information is sought with certainty. However, if a set of states
is not perfectly distinguishable then the {}``which state'' information
may still be obtained locally with some nonzero probability $p>0$
and this is what conclusive state discrimination seeks to achieve.
We say that a set of quantum states is conclusively locally distinguishable
if and only if with some nonzero probability \emph{every} state can
be correctly identified by LOCC \cite{Bandyo+Walgate-2009,Chefles2004}.
This means that there is a LOCC protocol whereby with some nonzero
probability it can be determined in which state the system was \emph{certainly}
prepared. It is clear that this notion can be naturally extended to
the case of many copies. 

We say that a given set of bipartite (or multipartite) orthogonal
quantum states is said to be \emph{maximally indistinguishable} if
and only if the set is not conclusively locally distinguishable in
the many copy limit. The rationale behind is simple: if a set of states
is not conclusively locally distinguishable with multiple copies,
no matter how many are available, then the only way to correctly identify
\emph{every} state, either with a nonzero probability, or with certainty,
must involve using additional resources like entanglement. For any
other kind of ensemble the {}``which state'' information can always
be obtained in the many copy limit, either with probability one, or
with some nonzero probability and therefore these ensembles are relatively
more distinguishable. 

We begin by proving our result for pure states. 

\begin{thm} Any set of $N$ orthogonal pure states of a composite
quantum system can be perfectly distinguished by LOCC with at most
$(N-1)$ copies. \end{thm}
\begin{proof}
We prove it for bipartite case only. The extension to the multipartite
case is straightforward. Suppose the following set of orthogonal pure
states $\left\{ |\psi_{1}\rangle,|\psi_{2}\rangle,...,|\psi_{N}\rangle\right\} \in\mathcal{H}_{A}\otimes\mathcal{H}_{B}$
is not perfectly distinguishable by LOCC for a single copy and also
suppose that a finite number of copies of the unknown state are available.
The strategy is to measure each copy separately, one after the other.
As we will show, every round of measurement performed on a single
copy always succeeds in eliminating at least one state. That is, after
$k$ rounds of measurements on $k$ copies, at least $k$ states get
eliminated. Thus, the states can be perfectly distinguished using
at most $\left(N-1\right)$ copies. 

Using the result of Walgate et al \cite{Walgate-2000} we write the
states (unnormalized) as,\begin{eqnarray}
|\psi_{1}\rangle & = & \sum_{i=1}^{d_{A}}|i\rangle_{A}\otimes|\phi_{i}\rangle_{B},\label{proof-N-pure-eq1}\\
|\psi_{2}\rangle & = & \sum_{i=1}^{d_{A}}|i\rangle_{A}\otimes|\phi_{i}^{\perp}\rangle_{B},\label{proof-N-pure-eq2}\end{eqnarray}
where, $\langle\phi_{i}^{\perp}|\phi_{i}\rangle=0\;\forall i$. The
remaining $(N-2)$ states can be written as \begin{equation}
|\psi_{n}\rangle=\sum_{i=1}^{d_{A}}|i\rangle_{A}\otimes|\chi_{i}^{n}\rangle_{B}\;:n=3,...,N\label{proof-N-pure-eq3}\end{equation}
We assume that for a given $n$, the states $|\chi_{i}^{n}\rangle\,:i=1,...,d_{A}$
are not orthogonal to $\left\{ |\phi_{i}\rangle,|\phi_{i}^{\perp}\rangle\;:i=1,...,d_{A}\right\} $.
Also $\langle\phi_{i}^{\perp}|\phi_{j}^{\perp}\rangle\neq0$ and $\langle\phi_{i}|\phi_{j}\rangle\neq0$
for all $i$ and $j$. Notice that the first two states are written
in a canonical form such that they can always be distinguished by
LOCC \cite{Walgate-2000}. Now, for every $i\;:i=1,...,d_{A}$ consider
an orthogonal basis on Bob's side: $\mathcal{B}_{i}=\left\{ |\phi_{i}\rangle,|\phi_{i}^{\perp}\rangle,|\eta_{p}^{i}\rangle\;:p=3,...,d_{B}\right\} $.
One can therefore write the states $|\chi_{i}^{n}\rangle$ as,\begin{equation}
|\chi_{i}^{n}\rangle=a_{i}^{n}|\phi_{i}\rangle+b_{i}^{n}|\phi_{i}^{\perp}\rangle+\sum_{p=3}^{d_{B}}c_{i,p}^{n}|\eta_{p}^{i}\rangle\;:\forall i,n\label{proof-N-pure-eq4}\end{equation}
where we assume that the coefficients $\left\{ a_{i}^{n},b_{i}^{n},c_{i,p}^{n}\right\} $
are all non-zero. The rest of the proof goes as follows. Alice performs
a measurement in the basis $\left\{ |i\rangle_{A}\;:i=1,...,d_{A}\right\} $.
For the $k^{th}$ outcome on Alice's side, Bob performs a measurement
in the basis $\mathcal{B}_{k}$. If neither $|\phi_{k}\rangle$ nor
$|\phi_{k}^{\perp}\rangle$ is obtained then the measurement eliminates
the first two states from contention. However, if either $|\phi_{k}\rangle$
or $|\phi_{k}^{\perp}\rangle$ is the outcome, then only one state
gets eliminated. In this case, the state which is eliminated is either
$|\psi_{2}\rangle$ (if the the outcome is $|\phi_{k}\rangle$) or
$|\psi_{1}\rangle$ (if the outcome is $|\phi_{k}^{\perp}\rangle$.
Thus $(N-1)$ states are left to distinguish after round one {[}or
$(N-2)$ depending on Bob's outcome{]} performed on the first copy. 

The same strategy is followed to perform next round of measurements
on the second copy to distinguish the remaining $\left(N-1\right)$
states (or $(N-2)$). Once again, we select any two states and write
them in the canonical form of Eqns.$\,$(\ref{proof-N-pure-eq1})
and (\ref{proof-N-pure-eq2}) and the remaining states in the form
of Eq.$\,$(\ref{proof-N-pure-eq3}). Alice and Bob perform the suitable
measurements to eliminate at least one state and the protocol continues.
In the worst case scenario, only one state gets eliminated in every
round at the expense of one copy of the state whose identity we are
trying to determine. Thus the states $\left\{ |\psi_{1}\rangle,|\psi_{2}\rangle,...,|\psi_{N}\rangle\right\} $
can be perfectly distinguished by LOCC with at most $(N-1)$ copies. 
\end{proof}
Before we prove the results for mixed states it is necessary to review
the relevant notions of LOCC distinguishability. We assume that we
are dealing with finite dimensional quantum systems. 

A positive semidefinite operator acting on $\mathcal{H}=\mathcal{H}_{A}\otimes\mathcal{H}_{B}$
is said to be separable if its normalized form is a separable quantum
state. A separable measurement $\Pi=\left\{ \Pi_{1},\Pi_{2},\cdots,\Pi_{n}\right\} $
on $\mathcal{H}$ is a POVM satisfying $\sum_{i=1}^{n}\Pi_{i}=\mathcal{I_{H}}$,
where $\Pi_{i}$ is a separable, positive semidefinite operator for
every $i$. Note that any measurement realized by LOCC is separable,
while the converse is not true \cite{Bennett-I-99}. To show that
a set of states is locally distinguishable it is necessary to demonstrate
it by an explicit protocol because states that are perfectly distinguishable
by separable operations but not by LOCC do exist \cite{Bennett-I-99}.
On the other hand, to show that a set of states is not perfectly distinguishable
by LOCC it suffices to show that the necessary conditions \cite{Hayashi-etal-2006,Watrous-2005,DuanFeng2007}
stated in the following proposition are violated. 

\begin{prop} If a set of orthogonal quantum states $\{\rho_{1},\rho_{2},...,\rho_{n}\}$
is perfectly distinguishable by LOCC then it is necessary that there
exists a separable POVM $\Pi=\left\{ \Pi_{1},\Pi_{2},\cdots,\Pi_{n}\right\} $
such that \begin{equation}
\mbox{Tr}\left(\Pi_{i}\rho_{j}\right)=\delta_{ij}.\label{prop1-eq-1}\end{equation}
 \end{prop} 

Recall that a set of states is said to be conclusively locally distinguishable
if and only if with some nonzero probability \emph{$p>0$ every} state
can be correctly identified by LOCC \cite{Bandyo+Walgate-2009,Chefles2004}.
The following proposition provides a necessary condition \cite{Chefles2004,Bandyo+Walgate-2009}.

\begin{prop} If a set of orthogonal quantum states $\{\rho_{1},\rho_{2},...,\rho_{n}\}$
is conclusively locally distinguishable by LOCC then it is necessary
that for every $i$ there exists a product state $|\phi_{i}\rangle$
such that $\forall j\neq i$ $\langle\phi_{i}|\rho_{j}|\phi_{i}\rangle=0$
and $\langle\phi_{i}|\rho_{i}|\phi_{i}\rangle\neq0$. \end{prop} 

\begin{remark} To prove that a set of orthogonal states is not conclusively
locally distinguishable, it suffices to show that there exists at
least one state which cannot be correctly identified with a nonzero
probability by LOCC. \end{remark}

It is obvious that if a set is not conclusively locally distinguishable
then it cannot be perfectly distinguished. The converse is not true
in general. For instance, the set of product states exhibiting nonlocality
without entanglement \cite{Bennett-I-99} is not perfectly distinguishable
by LOCC but clearly conclusively locally distinguishable in accordance
with the above proposition. On the other hand, the set of four Bell
states in $2\otimes2$, is neither perfectly distinguishable by LOCC
nor conclusively locally distinguishable. Note that the states in
the aforementioned examples are all perfectly distinguishable in the
many copy limit by Theorem 1. 

In what follows we give examples of bipartite orthogonal density matrices
that cannot be conclusively distinguished with multiple copies by
LOCC which immediately implies that they cannot be perfectly distinguished
either. The examples are obtained using the concept of unextendible
product basis \cite{Bennett-II-99,Divincenzo-2003}. 

A UPB is an orthogonal product basis on $\mathcal{H}=\mathcal{H}_{A}\otimes\mathcal{H}_{B}$
spanning a subspace $\mathcal{S}$ of $\mathcal{H}$ such that its
complementary subspace $\mathcal{S}^{\perp}$ contains no product
state. Note that by definition $\mathcal{S}\oplus\mathcal{S}^{\perp}=\mathcal{H}$.
Let $\mathcal{S}$ be a UPB on $\mathcal{H}=\mathcal{H}_{A}\otimes\mathcal{H}_{B}$
and $\mathcal{S}^{\perp}$ be its complementary subspace. Let $\sigma$
be the normalized projector onto the subspaces $\mathcal{S}$, and
$\rho$ be any quantum state belonging to the subspace $\mathcal{S}^{\perp}$. 

\begin{lem} The orthogonal density matrices $\left\{ \sigma,\rho\right\} $
are not perfectly distinguishable by LOCC. Moreover they are not conclusively
locally distinguishable. \end{lem}
\begin{proof}
It suffices to prove that the states are not conclusively locally
distinguishable. We will show that the state $\rho$ cannot be correctly
identified with nonzero probability by LOCC. Suppose there exists
a LOCC protocol whereby $\rho$ can be conclusively distinguished.
Then there exists a product state $|\phi\rangle$ in accordance with
Proposition 2 such that the following two equations must be satisfied:
\begin{eqnarray}
\langle\phi|\rho|\phi\rangle & \neq & 0\label{lemma-1-eq-1}\\
\langle\phi|\sigma|\phi\rangle & = & 0\label{lemma-1-eq-2}\end{eqnarray}
 Noting that $\sigma$ is the normalized projector onto the subspace
$\mathcal{S}$, from Eq.$\,$(\ref{lemma-1-eq-2}) it follows that
the product state $|\phi\rangle$ must lie in the subspace $\mathcal{S}^{\perp}$.
This is in contradiction with the fact that $\mathcal{S}^{\perp}$
contains no product state. This concludes the proof. 
\end{proof}
We now consider local distinguishability of the states $\left\{ \sigma,\rho\right\} $
in the many copy scenario. That is, we would like to know whether
the states $\left\{ \sigma{}^{\otimes n},\rho^{\otimes n}\right\} \in\mathcal{H}^{\otimes n}=\mathcal{H}_{A}^{\otimes n}\otimes\mathcal{H}_{B}^{\otimes n}$
can be perfectly distinguished by LOCC for some $n\geq2$. To begin
with we note that the tensor product of any two bipartite UPBs is
again a UPB. 

\begin{lem} \cite{Divincenzo-2003} Let $\mathcal{S}_{1}$ and $\mathcal{S}_{2}$
are UPBs on $\mathcal{H}_{1}$ and $\mathcal{H}_{2}$ respectively.
Then $\mathcal{S}_{1}\otimes\mathcal{S}_{2}$ ia a UPB on $\mathcal{H}_{1}\otimes\mathcal{H}_{2}$.
\end{lem} 

Therefore, $\mathcal{S}^{\otimes n}$ is also a UPB on $\mathcal{H}^{\otimes n}=\mathcal{H}_{A}^{\otimes n}\otimes\mathcal{H}_{B}^{\otimes n}$
which in turn implies that its orthogonal complement $\left(\mathcal{S}^{\otimes n}\right)^{\perp}$
does not contain any product state. 

\begin{lem} The orthogonal density matrices $\left\{ \sigma^{\otimes n},\rho^{\otimes n}\right\} $
are not conclusively distinguishable by LOCC for any finite $n$.
\end{lem}
\begin{proof}
We first note that $\sigma^{\otimes n}$ is the normalized projector
onto the UPB subspace $\mathcal{S}^{\otimes n}$. Then the orthogonality
condition $\sigma^{\otimes n}\perp\rho^{\otimes n}$ implies that
 the state $\rho^{\otimes n}$ belongs to the subspace $\left(\mathcal{S}^{\otimes n}\right)^{\perp}$.
Now suppose that the states can be conclusively locally distinguished.
Therefore, both $\sigma^{\otimes n}$ and $\rho^{\otimes n}$ can
be correctly identified with some nonzero probability by LOCC. For
$\rho^{\otimes n}$ this means, there exist a product state $|\eta\rangle$
in accordance with Proposition 2 satisfying the following equations:
\begin{eqnarray}
\langle\eta|\rho^{\otimes n}|\eta\rangle & \neq & 0\label{lemma-3-eq-1}\\
\langle\eta|\sigma^{\otimes n}|\eta\rangle & = & 0\label{lemma-3-eq-2}\end{eqnarray}
 By noting that $\sigma^{\otimes n}$ is the normalized projector
onto the subspace $\mathcal{S}^{\otimes n}$, from Eq.$\,$(\ref{lemma-3-eq-2})
it follows that the product state $|\eta\rangle$ must belong to the
subspace $\left(\mathcal{S}^{\otimes n}\right)^{\perp}$. However
the subspace $\left(\mathcal{S}^{\otimes n}\right)^{\perp}$ being
the orthogonal complement of the UPB subspace $\mathcal{S}^{\otimes n}$
does not contain any product state. This holds for any finite $n$.
Hence the proof. 
\end{proof}
Let us emphasize that the state $\sigma^{\otimes n}$ can be correctly
identified with some nonzero probability by LOCC. However, if the
system was prepared in the state $\rho^{\otimes n}$, then there is
no LOCC protocol whereby with some nonzero probability it can be determined
that the system was certainly prepared in $\rho^{\otimes n}$. As
is evident, the concept of UPB and its very special properties are
extremely crucial for our results. 

\begin{thm} Any bipartite orthogonal ensemble which contains $\sigma$,
the normalized projector onto a UPB subspace, is conclusively locally
indistinguishable in the many copy limit. \end{thm} 

The proof follows from Lemma 3. Notice that in any such ensemble the
state $\sigma$, in the single copy case or $\sigma^{\otimes n}$,
in the many copy case acts as a {}``lock down'' state robbing off
the possibility to correctly identify, even conclusively, any other
state in the ensemble but itself. 

Qualitatively speaking, any set of orthogonal quantum states therefore
must belong to either of the following three classes based on their
local distinguishability in the \emph{many copy} limit: (a) perfectly
distinguishable, (b) conclusively locally distinguishable, and (c)
not conclusively locally distinguishable. As Theorem 2 shows, there
exist bipartite orthogonal quantum states that are not conclusively
distinguishable by LOCC with many copies. Therefore, the only conceivable
way to extract the which state information must involve using entanglement
as an additional resource. Clearly such ensembles are less distinguishable
than other orthogonal ensembles for they are either perfectly distinguishable
or conclusively distinguishable when multiple copies are available,
and therefore allows us to determine the state of the system either
with certainty or with some nonzero probability. We therefore say
that ensembles of type (c) are maximally indistinguishable. From nonlocality
point of view, it is evident that such ensembles are {}``maximally
nonlocal''. 

To conclude, we have shown that nonlocality of quantum information
is perhaps more subtle than previously thought. In earlier examples
\cite{Bennett-I-99,Horodecki-2003} nonlocal nature of quantum information
was demonstrated within the set up of LOCC distinguishability of quantum
states, where the problem was defined with the assumption that only
a single copy of the unknown state is available. By removing this
constraint on the number of copies we have been able to show that
pure states cease to be nonlocal. On the other hand orthogonal mixed
states are found to be robust carriers of nonlocality because they
may not be perfectly distinguishable even in the many copy limit.
We have also argued that a set of orthogonal quantum states can be
appropriately called {}``maximally nonlocal'' iff the set is maximally
indistinguishable. 

Two issues are of immediate interest. First, it would be very interesting
to construct examples of maximally nonlocal orthogonal states without
invoking UPBs. If such an example could be found we will get a better
understanding of what classes of quantum states admit this specific
property of maximal nonlocality. Secondly, an example of a set of
orthogonal states of type (b) will help us to complete the general
picture.


\begin{thebibliography}{23}
\bibitem{Bennett-I-99} C. H. Bennett, D. P. DiVincenzo, C. A. Fuchs,
T. Mor, E. rains, P. W. Shor, J. A. Smolin, and W. K. Wootters, ,
Physical Review A \textbf{59}, 1070 (1999).

\bibitem{Bennett-II-99} C. H. Bennett, D. P. DiVincenzo, T. Mor,
P. W. Shor, J. A. Smolin, and B. M. Terhal, Physical Review Letters
\textbf{82}, 5385 (1999).

\bibitem{Divincenzo-2003} David P. DiVincenzo, Tal Mor, Peter W.
Shor, John A. Smolin, Barbara M. Terhal, Comm. Math. Phys. \textbf{238},
379 (2003). 

\bibitem{Walgate-2000} J. Walgate, A. J. Short, L. Hardy, and V.
Vedral, Physical Review Letters \textbf{85}, 4972 (2000).

\bibitem{Walgate-2002}J. Walgate and L. Hardy, Physical Review Letters
\textbf{89}, 147901 (2002). 

\bibitem{Ghosh-2001} S. Ghosh, G. Kar, A. Roy, A. Sen (De), and U.
Sen, Physical Review Letters \textbf{87}, 277902 (2001). 

\bibitem{Ghosh-2002}S. Ghosh, G. Kar, A. Roy, D. Sarkar, A. Sen (De),
and U. Sen, Physical Review A \textbf{65}, 062307 (2002). 

\bibitem{Ghosh-2004} S. Ghosh, G. Kar, A. Roy, and D. Sarkar, Physical
Review A \textbf{70}, 022304 (2004). 

\bibitem{Chen+Li} P. -X. Chen and C. -Z. Li, Physical Review A \textbf{68},
062107 (2003). 

\bibitem{Fan-2004}H. Fan, Physical Review Letters \textbf{92}, 177905
(2004).

\bibitem{Horodecki-2003}M. Horodecki, A. Sen (De), U. Sen, and K.
Horodecki, Physical Review Letters \textbf{90}, 047902 (2003). 

\bibitem{Nathanson-2005} M. Nathanson, Journal of Mathematical Physics
\textbf{46}, 062103 (2005). 

\bibitem{Watrous-2005} J. Watrous, Physical Review Letters \textbf{95},
080505 (2005).

\bibitem{Hayashi-etal-2006} M. Hayashi, D. Markham, M. Murao, M.
Owari, and S. Virmani, Physical Review Letters \textbf{96}, 040501
(2006). 

\bibitem{Duan2007}R. Y. Duan, Y. Feng, Z. F. Ji, and M. S. Ying,
Physical Review Letters\textbf{ 98}, 230502 (2007). 

\bibitem{DuanFeng2007} R. Y. Duan, Y. Feng, Y. Xin, and M. S. Ying,
IEEE Trans. Inform. Theory \textbf{55}, 1320 (2009). 

\bibitem{Bandyo+Walgate-2009} S. Bandyopadhyay and J. Walgate, Local
distinguishability of any three quantum states, J. Phys. A: Math.
Theor. \textbf{42} 072002 (2009).

\bibitem{Bandyopadhyay2010} S. Bandyopadhyay, Physical Review A \textbf{81},
022327 (2010). 

\bibitem{Virmanietal2001} S. Virmani, M. F. Sacchi, M. B. Plenio,
and D. Markham, Physics Letters A \textbf{288}, 62 (2001).

\bibitem{ChenYang2002} Y. -X. Chen, and D. Yang, Physical Review
A \textbf{65}, 022320 (2002). 

\bibitem{Chefles2004} A. Chefles, Physical Review A \textbf{69},
050307 (2004). 

\bibitem{Cohen2007} S. M. Cohen, Physical Review A \textbf{75}, 052313
(2007). 

\bibitem{Bagan-2010} J. Calsamiglia, J. I. de Vicente, R. Munoz-Tapia,
E. Bagan, Physical Review Letters \textbf{105}, 080504 (2010). 
\end{thebibliography}
\end{document}